\documentclass[conference,10pt]{IEEEtran}
\usepackage{amsmath,amssymb,eucal,graphicx}%,doublespace}
\usepackage{pst-plot}
\usepackage{pst-node}
\usepackage{epsfig}
\usepackage{amsthm}
\usepackage{amsfonts}
\usepackage{epsfig}
\usepackage{float}
\usepackage{pstricks}
\usepackage{color}
\usepackage{multicol}
\usepackage{mathcomSTEP}
\usepackage{subfigure,url}

\def \EE{\mathbb E}

\def \RR{\mathbb R}

\newtheorem{thm}{Theorem}[section]

\newtheorem{lem}[thm]{Lemma}

\theoremstyle{remark}

\theoremstyle{definition}

\begin{document}

\title{New Results on the Capacity of the Gaussian Cognitive Interference Channel}
%\title{New Capacity Results for the Gaussian Cognitive Interference Channel}

\author{
\IEEEauthorblockN{Stefano Rini, Daniela Tuninetti, and Natasha Devroye\\}%\authorrefmark{1}
\medskip
\IEEEauthorblockA {%
Department of Electrical and Computer Engineering\\
University of Illinois at Chicago\\
Email: \{srini2, danielat, devroye\}@uic.edu}
\thanks{The work of S. Rini and D. Tuninetti was partially funded by NSF under award 0643954.}
}

\maketitle
\begin{abstract}
The capacity of the two-user Gaussian cognitive interference channel,
a variation of the classical interference channel where
one of the transmitters has knowledge of both messages,
is known in several parameter regimes but remains unknown in general.
In this paper, we consider the following achievable scheme:
the cognitive transmitter pre-codes its message against the
interference created at its intended receiver by the primary
user, and the cognitive receiver only decodes its intended message,
similar to the optimal scheme for ``weak interference'';
the primary decoder decodes both messages, similar to the
optimal scheme for ``very strong interference''.
Although the cognitive message is pre-coded against the
primary message, by decoding it, the primary receiver
obtains information about its own message, thereby improving
its rate.
We show:
(1) that this proposed scheme achieves capacity
in what we term the ``primary decodes cognitive'' regime, i.e.,
a subset of the ``strong interference'' regime that is
not included in the ``very strong interference'' regime
for which capacity was known;
(2) that this scheme is within one bit/s/Hz, or a factor
two, of capacity for a much larger set of parameters, thus
improving the best known constant gap result;
(3) we provide insights into the trade-off between
interference pre-coding at the cognitive encoder
and interference decoding at the primary receiver based on the
analysis of the approximate capacity results.
\end{abstract}

%%%%%%%%%%%%%%%%%%%%%%%%%%%%%%%%%%%%%%%%%%%%%%%%%%%%%%%%%%%
\section{Introduction}
%%%%%%%%%%%%%%%%%%%%%%%%%%%%%%%%%%%%%%%%%%%%%%%%%%%%%%%%%%%
\label{sec:intro}

The {\it cognitive interference channel}~\cite{devroye_IEEE} is a well studied
channel model inspired by the newfound abilities of cognitive radio technology and
its potential impact on spectral efficiency in wireless networks~\cite{goldsmith_survey}.
This channel model consists of a two-user interference channel, where
one transmitter-receiver pair is referred to as the {\em primary} user
and the other as the {\em cognitive} user. As opposed to the classical
interference channel, the cognitive transmitter has full
non-causal knowledge of both messages, idealizing the cognitive user's ability
to detect transmissions taking place in the network.\footnote{This channel
has also been termed the {\em interference channel with unidirectional
cooperation}~\cite{MaricGoldsmithKramerShamai07Eu} or the {\em interference channel with degraded
message set}~\cite{WuDegradedMessageSet}.}
The primary transmitter has knowledge of its own message only.

%\cite{JovicicViswanath06, WuDegradedMessageSet}
%\cite{MaricYatesKramer07}
%\cite{rini-capacity}

\smallskip
\noindent
{\bf Past work.}
The cognitive interference channel was first posed in an information theoretic
framework in~\cite{devroye_IEEE}, where an achievable rate region
(for discrete memoryless channels) and a broadcast-channel-based outer bound
(for Gaussian channels only) were proposed.
The first capacity results were determined in~\cite{JovicicViswanath06,WuDegradedMessageSet}
for channels with ``weak interference'' at the primary receiver.
In this regime, the cognitive transmitter
pre-codes its message against the interference created at its receiver by the
primary message, while the primary receiver treats the interference from
the cognitive transmitter as noise.
In contrast, in the regime of ``very strong interference'' at the primary receiver,
capacity is achieved
by simple superposition coding at the cognitive transmitter and
by having both receivers decode both messages~\cite{MaricYatesKramer07}.
For the general cognitive interference channel, the largest known achievable rate region is found
in~\cite{RTDjournal1} (where the inclusion of all previously proposed inner bounds is formally shown), while the tightest known outer bound is in~\cite{MaricGoldsmithKramerShamai07Eu} (based on a broadcast-channel-based technique from~\cite{NairGamal06}).
While in general the capacity region remains unknown, in Gaussian noise,
we recently demonstrated an achievable rate region which lies to
within 1.87~bits/s/Hz from capacity~\cite{rini2010capacity}.

\smallskip
\noindent
{\bf Contributions.}
In this paper, we:
\noindent
\\1) Derive a new capacity result by showing the achievability of an outer bound originally presented
     in \cite{MaricGoldsmithKramerShamai07Eu} for a subset of the ``strong interference'' regime.
\\2) Prove capacity to within one bit/s/Hz or to within a factor two in a large parameter region.
\\3) Highlight how non-perfect (partial) interference ``pre-cancellation'' at the cognitive user
     boosts the rate of the primary user and give numerical examples to quantify the rate improvements.

\smallskip
\noindent
{\bf Organization.}
The rest of the paper is organized as follows:
Section \ref{sec:ch.model} formally defines the channel model and
summarizes known results for the Gaussian channel;
Section \ref{sec:new results} proves the new capacity result;
Section \ref{sec:further results perfect dpc} analyzes the capacity achieving
scheme with perfect ``pre-cancellation'' of the interference and
shows two new approximate capacity results;
Section \ref{sec:further results NONperfect dpc} focuses on the achievable rates
with partial interference ``pre-cancellation''; and
Section \ref{sec:Conclusion and Future Work} concludes the paper.

%%%%%%%%%%%%%%%%%%%%%%%%%%%%%%%%%%%%%%%%%%%%%%%%%%%%%%%%%%%%
\section{Channel model and known results}
%%%%%%%%%%%%%%%%%%%%%%%%%%%%%%%%%%%%%%%%%%%%%%%%%%%%%%%%%%%%
\label{sec:ch.model}

A two-user InterFerence Channel (IFC) is a multi-terminal network with
two senders and two receivers.  The inputs $(X_1,X_2)$
are related to the outputs $(Y_1,Y_2)$ by a memoryless channel with
transition probability $P_{Y_1,Y_2|X_1,X_2}$.
Each transmitter $i$, $i\in\{1,2\}$, wishes
to communicate a message $W_i$ to receiver $i$.
Each message $W_i$, $i\in\{1,2\}$, is uniformly distributed on
$\{1, \ldots , 2^{N R_i}\}$, where $N$ represents the block-length
and $R_i$ the transmission rate. The two messages are independent.
In the classical IFC, the two transmitters operate independently.
Here we consider a variation of this set up by assuming that transmitter~1
(referred to as the {\em cognitive} user), in addition to its own message
also knows the message of transmitter~2 (referred to as the {\em primary} user).
A rate pair $(R_1,R_2)$ is achievable for a CIFC if there exists
a sequence of encoding functions
\begin{align*}
X_1^N = X_1^N(W_1, W_2), \quad
X_2^N = X_2^N(W_2),
\end{align*}
and a sequence of decoding  functions
\pp{
\widehat{W}_1  = \widehat{W}_1(Y_1^N),  \quad \widehat{W}_2  = \widehat{W}_2(Y_2^N),
}
such  that
\begin{align*}
\max_{i \in \{1,2 \}}\Pr\lsb  \widehat{W}_i \neq W_i  \rsb \to 0, \;\; \text{ as }N\to\infty.
\end{align*}
The capacity region is the convex closure of the set of achievable rates~\cite{ThomasCoverBook}.

\begin{figure}
\centering
\includegraphics[width=9 cm]{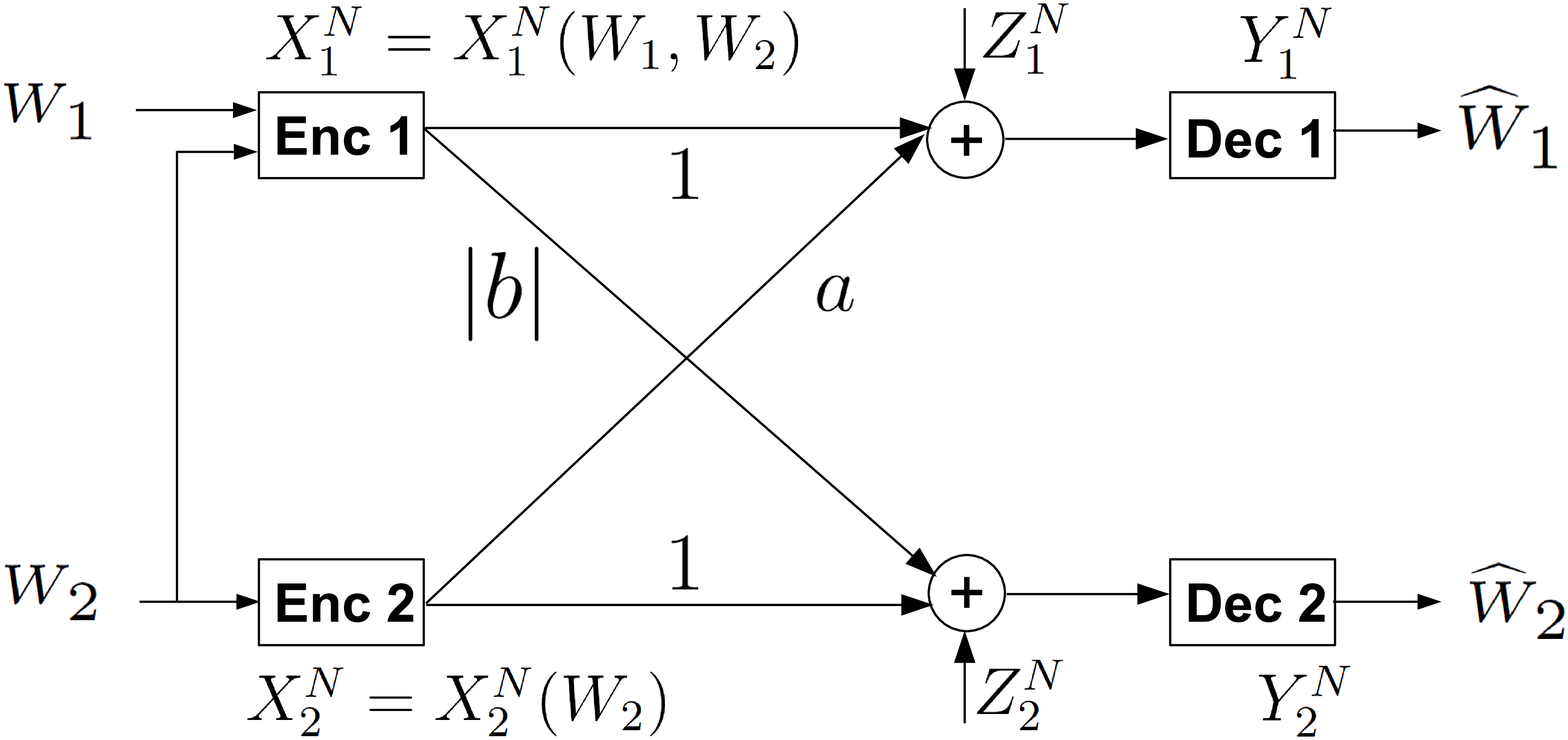}
\vspace{-.8 cm}
\caption{The G-CIFC.}
\label{fig:GeneralModel}
\end{figure}

%\subsection{The Gaussian CIFC}
%\label{sec:known results for gaussian}
A complex-valued Gaussian CIFC (G-CIFC) in \emph{standard form}, as
represented in Fig.~\ref{fig:GeneralModel}, has outputs:
\begin{align*}
Y_1 &= X_1+ a X_2+ Z_1,\\
Y_2 &= |b| X_1+ X_2+ Z_2,
\end{align*}
for $a,b \in \Cbb$, where the inputs are subject to the power constraint
$\E[|X_i|^2] \leq P_i$, $i \in \{1,2\}$,
and where the additive proper-complex Gaussian noises $Z_1$ and $Z_2$
have zero mean and unit variance. Without loss of generality~\cite{riniPrelim},
the phase of $b$ can be taken to be zero.

When $a=|b|^{-1}$, one channel output is a degraded version of the other output.
In  particular,
when $|b|\leq 1$, $Y_2$ is a degraded version of $Y_1$, and
when $|b|>    1$, $Y_1$ is a degraded version of $Y_2$.
We refer to this particular channel as the {\em degraded G-CIFC}.
%\begin{align*}
%Y_1 = X_1+\f 1 {|b|} X_2 +Z_1 \sim \f 1 {|b|} Y_2+ \sqrt{1-\f 1 {|b|^2}}\Zt_1,
%\end{align*}
%for $\Zt_1$ a proper-complex Gaussian random variable of zero mean,
%unit variance and independent of everything else.

The capacity region of the G-IFC is not known in general.
We next summarize known capacity results for the G-CIFC.
We define $\Ccal(x) := \log(1+x)$.

\begin{thm}{ \textbf{Outer bound of~\cite{rini2010capacity}.}}
\label{thm:unifying outer bound}
The capacity region of the G-CIFC is included in:
{%\small
\begin{subequations}
\begin{align}
R_1     & \leq \Ccal\lb \al P_1 \rb,
\label{eq:outer bound CIFC R1}\\
R_2     & \leq \Ccal\lb |b|^2 P_1 +P_2 + 2  \sqrt{\alb |b|^2  P_1 P_2} \rb,
\label{eq:outer bound CIFC R2}\\
R_1+R_2 & \leq \Ccal\lb |b|^2 P_1 +P_2 + 2  \sqrt{\alb |b|^2  P_1 P_2} \rb\nonumber\\&\quad
+\log \left( \f{1+\max\{1,|b|^2\}\al P_1}{1+\al|b|^2P_1} \right),
\label{eq:outer bound CIFC R1+R2}
\end{align}
\label{eq:outer bound CIFC}
\end{subequations}
}
taken over the union of all $\al\in[0,1]$.
\end{thm}

\begin{thm}{\textbf{``Weak interference'' capacity of~\cite[Lemma 3.6]{WuDegradedMessageSet} and~\cite[Th. 4.1]{JovicicViswanath06}.}}
If
\ea{
|b|\leq 1, \;\; \text{ (``weak interference'' regime/condition)}
\label{eq:weak interference condition}
}
the outer bound of Th.\ref{thm:unifying outer bound} is tight.
%the capacity of the G-CIFC is:
%\begin{subequations}
%\begin{align}
%R_1 & \leq \Ccal(\al  P_1)
%\label{eq:capacity weak GCIFC R1}\\
%R_2 &\leq \Ccal \lb |b|^2 P_1 +P_2+2  \sqrt{\alb |b|^2  P_1 P_2}\rb - \Ccal(|b|^2 \al P_1 ),
%\label{eq:capacity weak GCIFC R2}
%\end{align}
%\label{eq:capacity weak GCIFC}
%\end{subequations}
%taken over the union of all $\al \in [0,1]$, with $\alb=1-\al$.
%\label{thm:capacity weak GCIFC}
\end{thm}

\begin{thm}{ \textbf{``Very strong interference'' capacity of~\cite[Thm. 6]{Maric_strong_interference} extended to complex-valued channels.}}
When
\ea{
(|a|^2-1) P_2-(|b|^2-1)P_1 -2 \big|a-|b|\big| \sqrt{P_1 P_2} \geq  0,&
\nonumber \\ \text{and} \ |b|> 1 \
\text{(``very strong interference'' regime/condition)}&
\label{eq:very strong interference condition}
}
the outer bound of Th.\ref{thm:unifying outer bound} is tight.
\label{thm: Very strong interference capacity}
\end{thm}
\begin{IEEEproof}
For complex-valued G-CIFC with $|b|>1$, the outer bound of Th.\ref{thm:unifying outer bound}
is achievable by the superposition-only scheme of~\cite{Maric_strong_interference} if
$I(Y_1; X_1,X_2) \geq I(Y_2; X_1, X_2)$ for all input distributions, that is, if
\begin{align}
&\EE[|Y_1|^2]-\EE[|Y_2|^2] = (|a|^2-1)P_2 -(|b|^2-1)P_1+ \nonumber
\\&+2 \sqrt{P_1 P_2}  (\Re\{a^* \rho\} - |b| \Re\{\rho \})\geq 0,
\quad \forall|\rho| \leq 1.
\label{eq:strong interference condition proof}
\end{align}
Let $\rho=|\rho| \eu^{\jj \phi_\rho}$
and $a=|a|  \eu^{\jj \phi_a}$. We have
\begin{align*}
\Re\{a^* \rho\} - |b| \Re\{\rho \}
  &= |\rho||a|\cos(\phi_\rho-\phi_a)
   - |\rho||b|\cos(\phi_\rho)
%\\&= |\rho|\Big [|a| \cos(\phi_a)-|b|\Big] \cos(\phi_\rho)
%   + |\rho|\Big [|a| \sin(\phi_a)    \Big] \sin(\phi_\rho)
%\\&= |\rho|\sqrt{( \Big[|a| \cos(\phi_a)-|b|\Big])^2+\Big[|a| \sin(\phi_a)\Big]^2}
%     \cos(\phi_\rho-x)
\\&= \big|a-|b|\big|\cdot |\rho|\cos(\phi),
\end{align*}
for some angle $\phi$. The condition in~\reff{eq:strong interference condition proof}
is thus verified for all $|\rho| \cos(\phi)\in[-1,+1]$ if it is verified for $|\rho|\cos(\phi) =-1$.
\end{IEEEproof}

\begin{thm}{ \textbf{Constant gap of \cite{rini2010capacity}. }}
\label{thm: constant gap}
The gap between the outer bound in Th.\ref{thm:unifying outer bound}
and the achievable scheme of \cite{rini2010capacity} is at most 1.87 bits/s/Hz.
\end{thm}

Fig.~\ref{fig:capacityPlotOld} shows the region where capacity is
known (shaded region) and where only constant gap results are
available (white region) in the $(a,|b|)$-plane
for $a \in \Rbb$ and $P_1=P_2$.

\begin{figure}
\centering
\includegraphics[width=9 cm]{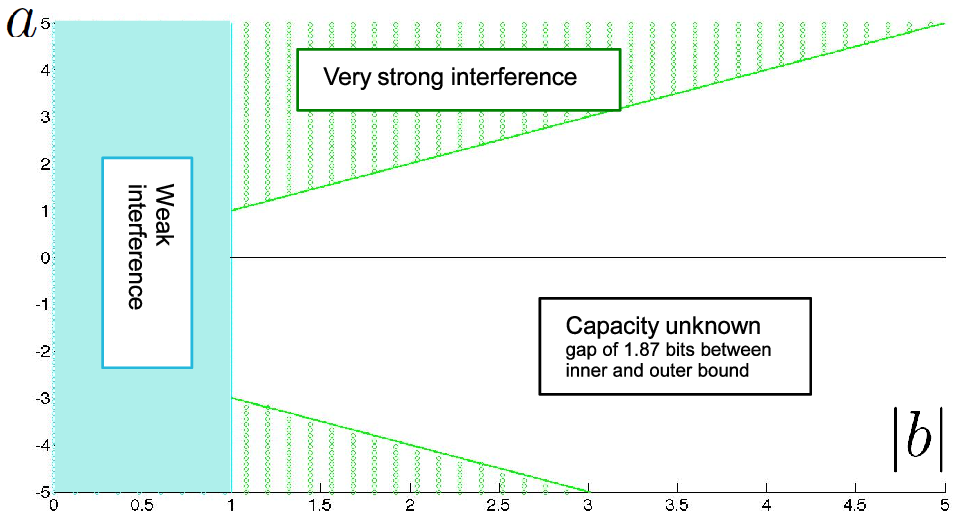}\\
\vspace{-1.2 cm}
\caption{A representation of the known capacity results
for the real-valued G-CIFC with $P_1=P_2$ and $(a,|b|)\in [-5,5]\times [0,5]$.
The ``weak interference'' regime has solid light blue fill, while
the ``very strong interference'' regime has dotted green lines.}
\label{fig:capacityPlotOld}
\end{figure}

\section{New capacity results}
\label{sec:new results}

We now show the achievability of the
outer bound of Th.\ref{thm:unifying outer bound}
for a subset of $|b|> 1$  not included in the ``very strong interference''
regime of Th.\ref{thm: Very strong interference capacity}.
We first present an achievable scheme and then derive the conditions under which it is capacity achieving.

\subsection{Achievable scheme}
\label{sec:scheme E}

Consider the achievable scheme of \cite{rini2009state} where only $R_{1c}$ and $R_{2pa}$
are non zero\footnote{This scheme can also be obtained as a special case of the region of \cite[Th.~1]{MaricGoldsmithKramerShamai07Eu} and \cite{cao2008interference}.}, that is:
\begin{itemize}
 \item the primary message  $W_2$ is encoded in $X_2^N$,
 \item the cognitive message $W_1$ is encoded in $U_{1c}^N$,
 \item $U_{1c}^N$ is Dirty Paper Coded (DPCed) against the interference
       created by $X_2^N$ at the cognitive receiver,
 \item $X_1^N$ is a function of $X_2^N$ and $U_{1c}^N$,
 \item the cognitive receiver decodes $U_{1c}^N$, and,
 \item the primary receiver jointly decodes $U_{1c}^N$ and $X_2^N$.
\end{itemize}
This scheme  achieves the following region:
\eas{
R_1 & \leq & I(Y_1; U_{1c})-I(U_{1c};X_2),
\label{eq:scheme X2 U1c DMC R1} \\
R_2 & \leq & I(Y_2, U_{1c};X_2) = I(Y_2; U_{1c},X_2) +\nonumber\\
   && -\Big(I(Y_2; U_{1c})-I(U_{1c};X_2)\Big),
\label{eq:scheme X2 U1c DMC R2}\\
R_1+R_2 & \leq & I(Y_2; U_{1c},X_2),
\label{eq:scheme X2 U1c DMC R1+R2}
}{\label{eq:scheme X2 U1c DMC}}
for any distribution $p_{U_{1c} X_1 X_2}$.

For the G-CIFC, we restrict our attention to the case:
\eas{
X_{1c} & \sim & \Ncal_{\Cbb}(0,\al P_1), \ \al \in [0,1], \\
X_2    & \sim & \Ncal_{\Cbb}(0,P_2)  \ , \  X_{1c}  \perp  X_2, \\
X_1    & =    & X_{1c}+ \sqrt{\f {\alb P_1}{P_2}} X_2,  \\
U_{1c} & =    & X_{1c}+ \la X_2, \ \la \in \Cbb,\\
Y_1 &=& X_{1c}+\lb a+\sqrt{\f {\alb P_1}{P_2}} \rb X_2 +Z_1, \\
Y_2 &=& |b| X_{1c}+\lb 1+\sqrt{\f {\alb |b|^2 P_1}{P_2}} \rb X_2 +Z_2.
}{\label{eq:primary decodes cognitive RV assignment}}
The parameter $\al$ denotes the fraction of power that encoder~1 employs
to transmit its own message versus the power to broadcast $X_2$.
For $\al=0$, transmitter~1 uses  all its power to broadcast $X_2$
as in a virtual Multiple Input Single Output (MISO) channel.
When $\al=1$, transmitter~1 utilizes all its power  to transmit $X_{1c}$.
With the choices in~\reff{eq:primary decodes cognitive RV assignment},
the region in \reff{eq:scheme X2 U1c DMC} becomes:
\eas{
R_1 &\leq & f\left(a+\sqrt{\f {\alb P_1}{P_2}},  1; \lambda\right),
\label{eq:scheme X2 U1c R1}\\
R_2 &\leq & \Ccal (P_2+|b|^2P_1+2 \sqrt{ \alb |b|^2 P_1 P_2})+\nonumber\\
       &&- f\left(\frac{1}{|b|}+\sqrt{\f {\alb P_1}{P_2}},\frac{1}{|b|^2}; \lambda\right),
%      &&- f\left(\frac{1}{|b|}\lb 1+\sqrt{\f {\alb |b|^2 P_1}{P_2}} \rb,\frac{1}{|b|^2}\right),
\label{eq:scheme X2 U1c R2}\\
R_1+R_2  &\leq & \Ccal (P_2+|b|^2P_1+2\sqrt{\alb |b|^2 P_1 P_2}),
\label{eq:scheme X2 U1c R1+R2}
}{\label{eq:scheme X2 U1c}}
for
\begin{align*}
&    f(h,\sigma^2;\lambda)
  \triangleq
   I(X_{1c}+ h X_2 + \sigma Z_1; U_{1c}) - I( U_{1c}; X_2)
\\&= \log \lb  \frac{\sgs+\al P_1}{\sgs
+\f{\al P_1   |h|^2P_2 }{\al P_1+|h|^2 P_2+\sgs}\left|\f{\la}{\la_{\rm Costa}(h,\sigma^2)}-1\right|^2}
\rb,
\end{align*}
where
\[
\la_{\rm Costa}(h,\sigma^2)
\triangleq  \f {\al P_1} {\al P_1 + \sigma^2} \ h.
\]
The parameter $\la$ controls how much of the interference created
by $X_2$ is ``pre-canceled'' and
\begin{align*}
f(h,\sigma^2;\la)
&\leq f\big(h,\sigma^2;\la_{\rm Costa}(h,\sigma^2)\big)
%\\&= I(X_{1c}+\sigma Z_1; X_{1c})
= \Ccal\lb \al P_1/\sigma^2\rb,
\end{align*}
i.e.,  $\la=\la_{\rm Costa}(h,\sigma^2)$
achieves the ``interference free'' rate as shown by Costa in~\cite{costa1983writing}.

\subsection{The pre-coding parameter $\lambda$}
\label{sec:discussion on lambda}

As mentioned before, the parameter $\la$ controls the amount of interference
``pre-cancellation'' achievable with DPC at transmitter~1.  With
$\la=0$, no DPC is performed at transmitter~1 and the interference
due to $X_2$ is treated as noise.
On the other hand, with
\[
\la
= \la_{\rm Costa}\lb a+\sqrt{\f {\alb P_1}{P_2}},\ 1 \rb
%= \f {\al P_1} {\al P_1 + 1} \ \lb a+\sqrt{\f {\alb P_1}{P_2}}\rb
\triangleq   \la_{\rm Costa \ 1},
\]
the interference due to $X_2$ at receiver~1 is completely ``pre-canceled'',
thus achieving the maximum possible rate $R_1$.
Different values of $\la$ are not usually investigated because,
as long as the interference is a nuisance (i.e., no node in the
network has information to extract from the interference), the best
is to completely ``pre-cancel'' it by using $\la = \la_{\rm Costa}(h,\sgs)$.

However, $\la$ influences not only the rate $R_1$ in \reff{eq:scheme X2 U1c R1},
but also the rate $R_2$ in \reff{eq:scheme X2 U1c R2}.  An interesting question
is whether $\la \not= \la_{\rm Costa \ 1}$, although it does not achieve
the largest possible $R_1$, would improve the
achievable region by sufficiently boosting the rate $R_2$.
We comment on this question later on.
At this point we make the following observation:
$R_1$ is a concave function in $\la$, symmetric around $\la=\la_{\rm Costa \ 1}$
and with a global maximum at $\la=\la_{\rm Costa \ 1}$, while
$R_2$ is a convex function in $\la$, symmetric around $\la=\la_{\rm Costa \ 2}$
and with a global minimum at  $\la=\la_{\rm Costa \ 2}$, where
\[
\la_{\rm Costa}\lb\frac{1}{|b|}+\sqrt{\f {\alb P_1}{P_2}},\ \frac{1}{|b|^2} \rb
%= \f {\al P_1} {\al P_1 + \frac{1}{|b|^2}} \lb\frac{1}{|b|}+\sqrt{\f {\alb P_1}{P_2}} \rb
\triangleq   \la_{\rm Costa \ 2}.
\]
Fig.~\ref{fig:MaricChannelAchievabilityV3} shows $R_1$
in \reff{eq:scheme X2 U1c R1} and $R_2$ in \reff{eq:scheme X2 U1c R2}
as a function of $\la\in\RR$,
for $P_1=P_2=6$, $b=\sqrt{2}$, $a=\sqrt{0.3}$ and $\al=0.5$.

\begin{figure}
\centering
\includegraphics[width=7 cm]{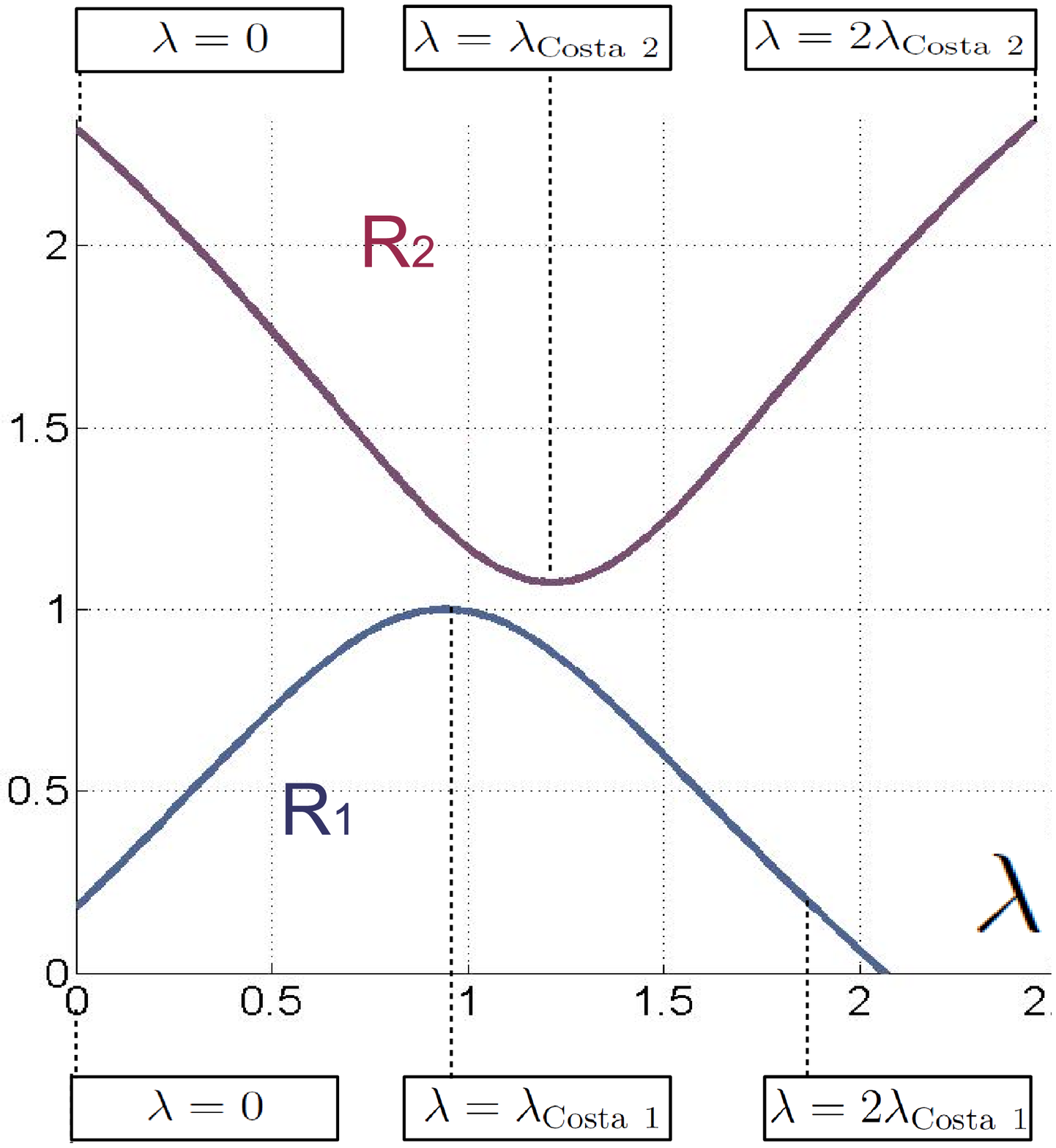}\\
\vspace{-1.8 cm}
\caption{The bound for $R_1$ in~\reff{eq:scheme X2 U1c R1} (bottom) and
the bound for $R_2$ in~\reff{eq:scheme X2 U1c R2} (top) as a function of $\la\in\RR$,
for $P_1=P_2=6$, $b=\sqrt{2}$ and $a=\sqrt{0.3}$ and $\al=0.5$.}
\label{fig:MaricChannelAchievabilityV3}
\end{figure}

\subsection{New ``strong interference'' capacity result}
\label{sec:new capacity}

We now determine the conditions under which the outer bound of
Th.\ref{thm:unifying outer bound} can be achieved by the region in \reff{eq:scheme X2 U1c}.
%The proof is established by equating the $\al$ parameters in \reff{eq:outer bound CIFC} and
%\reff{eq:scheme X2 U1c} and setting  $\la=\la_{\rm Costa}$ in \reff{eq:scheme X2 U1c}.

\begin{thm}{\bf{The ``primary decodes cognitive'' regime.}}
\label{thm:new capacity result}
When $|b|>1$ and
{\footnotesize
\begin{subequations}
\begin{align}
&P_2 |1-a|b||^2%\big(1+|a|^2-2\Re\{a\}|b| \big)
\geq (|b|^2-1)(1+P_1+|a|^2P_2) - P_1 P_2 \big|1-a |b| \big|^2, %\la=1
\label{eq:capacity condition 1}\\
&P_2 |1-a|b||^2%\big(1+|a|^2-2\Re\{a\}|b| \big)
\geq (|b|^2-1)(1+P_1+|a|^2P_2+ 2 \Re\{a\}\sqrt{P_1 P_2}), %\la=0
\label{eq:capacity condition 2}
\end{align}
\label{eq:capacity condition}
\end{subequations}
}\normalsize
the outer bound of Th.\ref{thm:unifying outer bound} is tight.
\end{thm}

The ``primary decodes cognitive'' regime, illustrated in Fig.~\ref{fig:New capacity results}
in the $(a,|b|)$-plane for $a \in \Rbb$ and $P_1=P_2=10$, covers parts of the ``strong interference''
regime where capacity was known to within 1.87 bits/s/Hz only. It also shows that the scheme
in~\reff{eq:scheme X2 U1c} is capacity achieving for part of the ``very strong interference'' region,
thus providing an alternative capacity achieving scheme to superposition coding~\cite{Maric_strong_interference}.

\begin{figure}
\centering
\includegraphics[width=10 cm ]{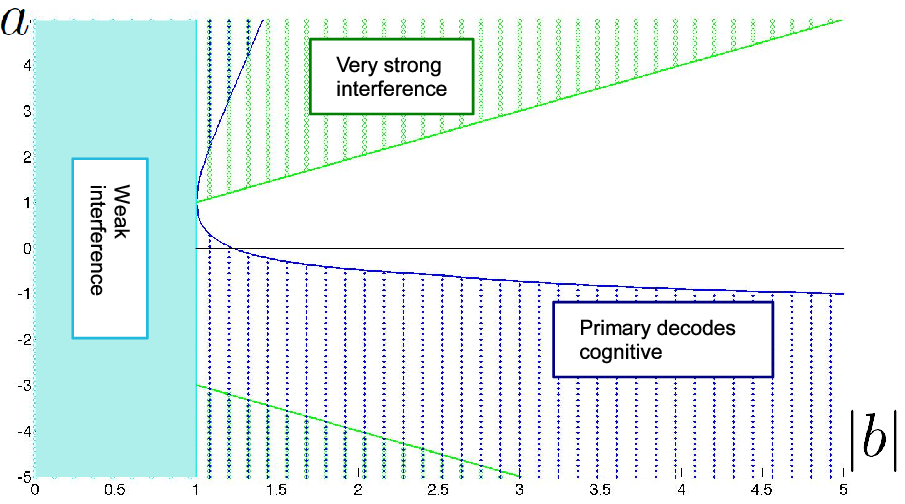}
\vspace{-1.2 cm}
\caption{A representation of the capacity result of Th.\ref{thm:new capacity result}
for a real-valued G-CIFC with $P_1=P_2=10$ and $(a,|b|)\in [-5,5]\times [0,5]$.
The ``weak interference'' regime has solid light blue fill,
the ``very strong interference'' regime has dotted green lines, while
the ``primary decodes cognitive'' region has dotted navy blue lines.}
\label{fig:New capacity results}
 \vspace{-.6 cm }
\end{figure}

%The following Maple code verified the conditions of the theorem
%for real-valued channels:
%\begin{verbatim}
%restart;

%# function
%lcosta := (h, sigma) -> alpha * P1 / (alpha * P1 + sigma^2) * h;

%
%f := (h, sigma, x) -> log( (alpha * P1 + sigma^2) / (sigma^2 + alpha * P1 * h^2 * P2 / (sigma^2 + alpha * P1 + h^2 * P2) *(x/lcosta(h, sigma)-1)^2) );

%# values
%h1     := a+sqrt((1-alpha)*P1/P2);
%sigma1 := 1;

%
%h2     := 1/b+sqrt((1-alpha)*P1/P2);
%sigma2 := 1/b;

%# We must solve g >= 0 <==> exp(g) >= 1
%g   := f(h1, sigma1, lcosta(h1, sigma1))-f(h2, sigma2, lcosta(h1, sigma1)):
%gs  := simplify(exp(g)-1):
%gsx := simplify(eval(gs,alpha=1-x^2))  assuming P1>0, P2>0, b>0, x>0, x<1;

%val0 := eval(gsx/(-1+x^2),x=0);
%val1 := eval(gsx/(-1+x^2),x=1);
%\end{verbatim}

\begin{proof}
The achievable scheme in \reff{eq:scheme X2 U1c} for $|b|>1$
and $\la=\la_{\rm Costa \ 1}$ achieves
\reff{eq:outer bound CIFC R1}=\reff{eq:scheme X2 U1c R1} and
\reff{eq:outer bound CIFC R1+R2}=\reff{eq:scheme X2 U1c R1+R2}
(and \reff{eq:outer bound CIFC R2} is redundant).
Therefore the outer bound of Th.\ref{thm:unifying outer bound}
is achievable when
(\reff{eq:scheme X2 U1c R1}+\reff{eq:scheme X2 U1c R2})$\geq$\reff{eq:outer bound CIFC R2},
that is when %$I(Y_2;U_{1c}) \leq I(Y_1; U_{1c})$, or equivalently
\begin{align}
&\Ccal(\al P_1) = f\Big(a+\sqrt{\f {\alb P_1}{P_2}},  1; \ \la_{\rm Costa \ 1}\Big)
\nonumber\\
&\geq
f\Big(\frac{1}{|b|}+\sqrt{\f {\alb P_1}{P_2}},\frac{1}{|b|^2}; \ \la_{\rm Costa \ 1}\Big),
\quad \forall \al\in[0,1].
\label{eq:proof capacity 1}
\end{align}
After some algebra, the inequality \reff{eq:proof capacity 1} can be shown to be equivalent to
\begin{align}
Q(\alpha)
&\triangleq P_2 \big|1-a |b| \big|^2 (\al  P_1 +1)-(|b|^2-1)\Big(P_1+|a|^2P_2+\nonumber \\
&+2 \Re\{a\} \sqrt{\alb P_1 P_2}+1\Big) \geq 0, \quad \forall \al\in(0,1].
\label{eq:inequality capacity proof}
\end{align}
The inequality in \reff{eq:inequality capacity proof}
is verified if $Q(0)\geq 0$ (which corresponds to the condition in~\reff{eq:capacity condition 2})
and $Q(1)\geq 0$ (which corresponds to the condition in~\reff{eq:capacity condition 1}).
\end{proof}

Previous capacity results for the G-CIFC imposed conditions on the channel
parameters that lent themselves well to ``natural'' interpretations. For example,
the ``weak interference'' condition $I(Y_1; X_1|X_2) \geq I(Y_2; X_1|X_2)$
suggests that decoding $X_1$ at receiver~2, even after having decoded
the intended message in $X_2$, would constrain the rate $R_1$
too much, thus preventing it from achieving the interference-free rate
in~\reff{eq:outer bound CIFC R1}.
The ``very strong interference'' condition $I(Y_1; X_1,X_2) \geq I(Y_2;X_1,X_2)$
suggests that requiring receiver~1 to decode both messages should not prevent
achieving the maximum sum-rate at receiver~2 given by~\reff{eq:outer bound CIFC R1+R2}.
A similar intuition about the new ``primary decodes cognitive'' capacity condition
in~\reff{eq:capacity condition} does not emerge from the proof of Th.\ref{thm:new capacity result}.
For this reason we devote the next sections to analyze the achievable
scheme of Section~\ref{sec:scheme E} in more detail. We shall consider the case of
perfect interference ``pre-cancellation'' at receiver~1 (i.e.,  $\la=\la_{\rm Costa \ 1}$)
and partial (or non-perfect) interference ``pre-cancellation'' at receiver~1  (i.e.,  $\la\not=\la_{\rm Costa \ 1}$)
separately.

\section{Perfect interference ``pre-cancellation''}
\label{sec:further results perfect dpc}

We next establish approximate capacity results by comparing the achievable performance
of the scheme in Section \ref{sec:scheme E} with perfect interference ``pre-cancellation''
at the cognitive receiver (i.e.,  $\la=\la_{\rm Costa \ 1}$)
to the outer bound in Th.\ref{thm:unifying outer bound}.
In particular, we show a parameter range where capacity can be achieved within
one bit/s/Hz, or within a factor two.

\subsection{Approximate capacity results}
Assume $\la=\la_{\rm Costa \ 1}$.  For $\al=0$,
the achievable scheme in~\reff{eq:scheme X2 U1c} is optimal
and achieves the MISO point $A$
in Fig.~\ref{fig:outerBoundPlotAllerton}.  Moreover,
when the  conditions in~\reff{eq:capacity condition} are verified,
the outer bound of Th.\ref{thm:unifying outer bound} is achievable
for all $\al\in[0,1]$.
In general, however the inner and outer bounds meet for some $\al \in [0,1]$.
In particular, from the proof of Th.\ref{thm:new capacity result},
it is possible to achieve the outer bound point for $\al=1$ whenever the
condition in~\reff{eq:capacity condition 1} holds.
We use this observation to reduce the additive gap of Th.~\ref{thm: constant gap}.
In addition we also provide a multiplicative gap.
The additive bound is effective at high SNR, where the difference between inner
and outer bound is small in comparison to the magnitude of the capacity region,
while the multiplicative bound is useful at low SNR, where the ratio between
inner and outer bound is a more indicative measure of their distance.

%bound the maximum distance between inner and outer bounds.
%when condition \reff{eq:capacity condition 1} holds.

\begin{figure}
\centering
 \includegraphics[width=9.3 cm ]{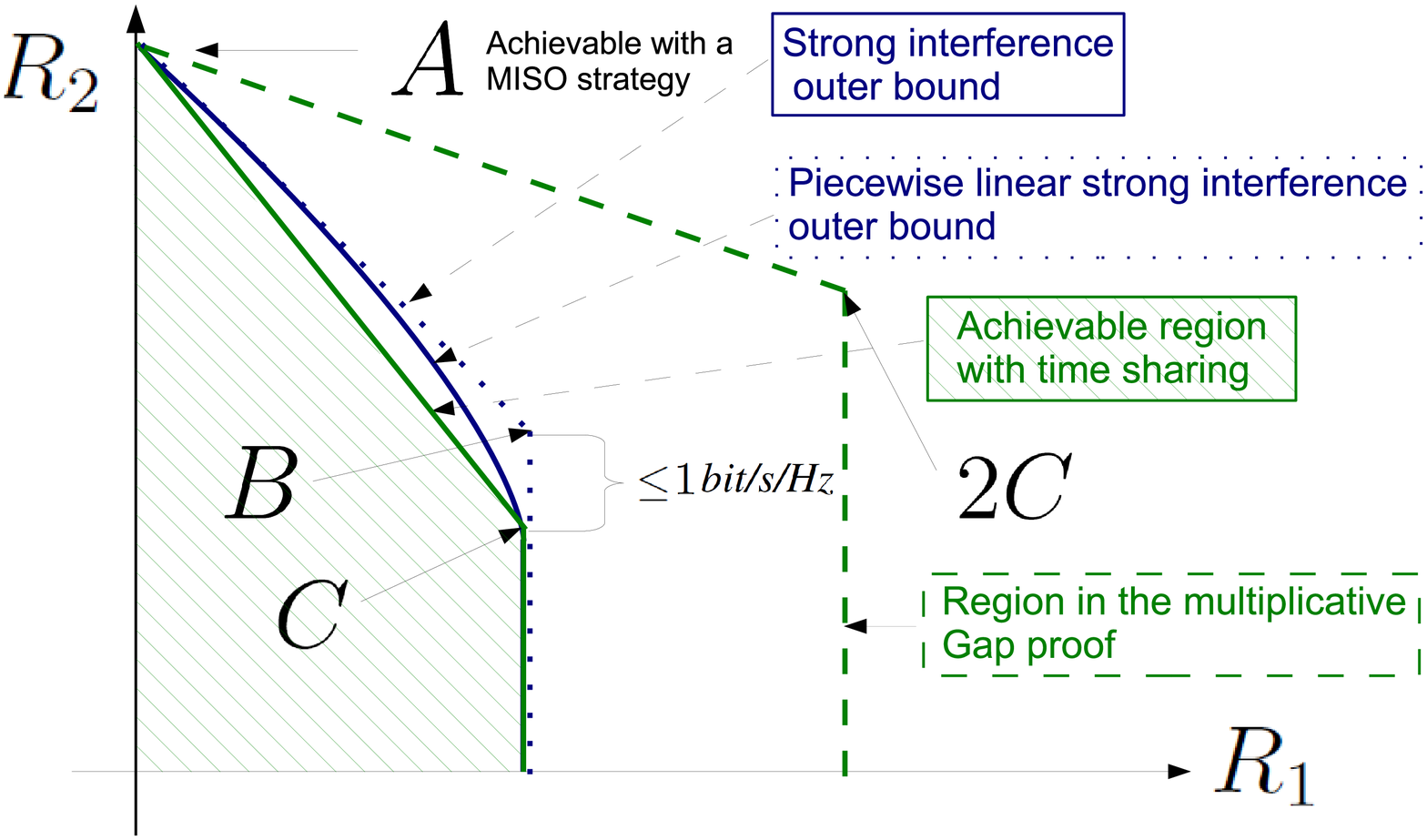}
 \vspace{-2 cm }
 \caption{A representation of the proof of Th.\ref{thm:banana shape gap}: the solid blue line is the outer bound in \reff{eq:outer bound CIFC} and the dotted blue line the outer bound in \reff{eq:piecewise SI}; the green hatched  region is the achievable  region obtained with time sharing between the points $A$ and $C$, while the green dotted line is the region in the converse of the multiplicative gap.}\label{fig:outerBoundPlotAllerton}
\end{figure}

\begin{thm}
\label{thm:banana shape gap}
When the condition in~\reff{eq:capacity condition 1} holds the outer bound of Th.\ref{thm:unifying outer bound}  is achievable
to within one bit/s/Hz or to within a factor two.
\end{thm}

\begin{IEEEproof}
By choosing the two values of $\al$ that maximize each bound
in~\reff{eq:outer bound CIFC}, we obtain that
the outer bound of Th.\ref{thm:unifying outer bound} for $|b|>1$
is contained in
\eas{
R_1 &\leq& \Ccal(P_1), \\
R_1+R_2 &\leq& \Ccal(( \sqrt{|b|^2 P_1}+\sqrt{P_2})^2),
}{\label{eq:piecewise SI}}
whose Pareto-optimal corner points (see Fig.~\ref{fig:outerBoundPlotAllerton}) are the MISO point
\[
A = \lb 0, \Ccal((\sqrt{|b|^2P_1}+\sqrt{P_2})^2 \rb
\]
and the point
\[
B= \lb \Ccal(P_1),  \Ccal(( \sqrt{|b|^2 P_1}+\sqrt{P_2})^2) - \Ccal(P_1) \rb.
\]
%We refer to this region as the ``piecewise linear strong interference'' outer bound since it is defined by two linear inequalities.
When condition~\reff{eq:capacity condition 1} is verified, it is possible to achieve the point of the outer bound
of Th.\ref{thm:unifying outer bound} corresponding to $\al=1$ given by (see Fig. \ref{fig:outerBoundPlotAllerton})
\[
C=\lb \Ccal(P_1),  \Ccal(|b|^2P_1+P_2)-\Ccal(P_1) \rb.
\]
Since points $B$ and $C$ have the same $R_1$-coordinate,
a constant additive gap  is readily shown by proving that
the difference of the $R_2$-coordinates is bounded.
We have:
\begin{align*}
&R_2^{(B)}-R_2^{(C)}
 =   \Ccal \lb \f{2 \sqrt{|b|^2P_1 P_2}} {1+ |b|^2P_1+ P_2} \rb
\\&\leq\Ccal \lb \sqrt{\f{P_2}{1+P_2}} \rb
  \leq\Ccal(1) = \log(2) = 1~\text{bit},
\end{align*}
where the largest gap is for $|b|^2P_1=P_2+1$.
%The function
%$\f{2 \sqrt{|b|^2P_1 P_2}} {1+ |b|^2P_1+ P_2}$ has a maximum at $|b|^2P_1=\sqrt {P_2+1}$ and thus we have
%\pp{
%\f{2 |b|^2P_1 P_2} {1+ |b|^2P_1+ P_2}  \leq \f{  \sqrt{P_2(P_2+1)}} {1+P_2} \leq 1,
%}
%from which we conclude that $R_2^{(B)}-R_2^{(C)} \leq \log(2)$.

For the multiplicative gap, it is sufficient to show that
$2(R_1^{(C)}+R_2^{(C)}) \geq R_1^{(B)}+R_2^{(B)}$ (see Fig.~\ref{fig:outerBoundPlotAllerton}),
that is
\begin{align*}
(1+|b|^2 P_1 + P_2)^2 \geq  1+|b|^2 P_1 +P_2 +2 \sqrt{|b|^2 P_1 P_2}  \ \iff\\
|\sqrt{|b|^2 P_1}-\sqrt{P_2}|^2 +(|b|^2 P_1+P_2)^2\geq 0,
\end{align*}
which is always verified.
\end{IEEEproof}

In Fig.~\ref{fig:BananaShape} we plot the region where the
condition in~\reff{eq:capacity condition 1}
holds for $a \in \Rbb$ and
for increasing values of $P_1=P_2=P$, that is, the region of $(a,|b|)$ pairs that satisfy
\ea{
P(P+1) |1-a |b||^2 \geq (|b|^2-1)(P+1+|a|^2 P).
\label{eq:capacity condition 1 P1=P2=P}
}
We notice that when $a |b| = 1$ (degraded channel),
the condition in~\reff{eq:capacity condition 1 P1=P2=P}
is satisfied by $a=|b|=1$ only.  As we increase $P$, the set of $(a,|b|)$ pairs
for which~\reff{eq:capacity condition 1 P1=P2=P}
is satisfied enlarges and converges to $a |b| \not= 1$.
In other words, for high-SNR, the scheme in Section~\ref{sec:scheme E}
with $\la=\la_{\rm Costa \ 1}$ is within one bit/s/Hz or a factor
two of capacity for all channel parameters, except
for the degraded channel.
However, for finite SNR, there exists a
region around the degraded line $a|b|=1$ for which the approximate
capacity result of Th.~\ref{thm:banana shape gap}
does not hold.
This is so because in this regime $\la_{\rm Costa \ 1} \approx \la_{\rm \ Costa \ 2}$
and thus the choice $\la=\la_{\rm Costa \ 1}$ gives a value for \reff{eq:scheme X2 U1c R2}
very close to its minimum
(see Fig.~\ref{fig:MaricChannelAchievabilityV3}) and  prevents the achievability of the point $C$.
This observation motivates the
study of imperfect, or partial, interference ``pre-cancellation''
($\la\not=\la_{\rm Costa \ 1}$) at the cognitive receiver.

\begin{figure}
\centering
 \includegraphics[width=9.5 cm ]{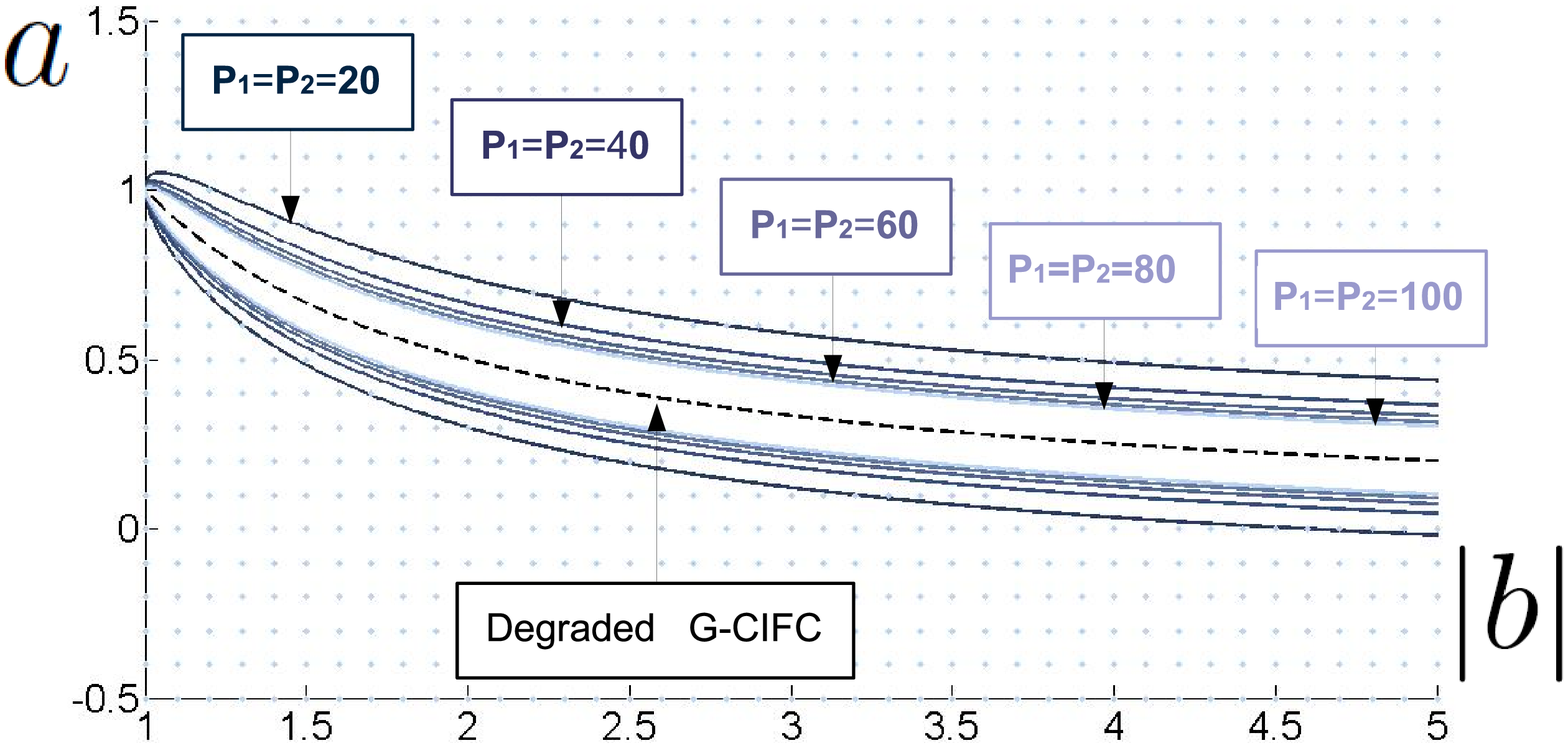}
 \vspace{-2.5 cm }
 \caption{The region where the condition in~\reff{eq:capacity condition 1 P1=P2=P}
 is verified, for different values of $P_1=P_2=P$. }
 \label{fig:BananaShape}
 \vspace{-.5 cm }
\end{figure}

\section{Partial interference ``pre-cancellation''}
\label{sec:further results NONperfect dpc}

\subsection{On the utility of partial interference ``pre-cancellation''}

In order to understand the rate improvements from partial interference ``pre-cancellation''
in the achievable scheme of Section~\ref{sec:scheme E}, we first revisit the gap result of
Th.\ref{thm: constant gap} from~\cite{rini2010capacity}.
In~\cite[Sec. IV.B]{rini2010capacity} we showed that in the subset
of the ``strong  interference'' regime where Th.~\ref{thm:banana shape gap} does not hold,
choosing $\la=\la_{\rm Costa \ 1}+\ep$, with $\ep$ an appropriate decreasing function
of the transmit powers, yields a gap to capacity
of at most 1.87 bits/s/Hz.  This choice of $\la$ can be interpreted as follows.
From~\reff{eq:scheme X2 U1c DMC R2}, we see that $U_{1c}$ plays the role of
side~information at receiver~2 when decoding $X_2$.  The DPC coefficient
$\la=\la_{\rm Costa \ 1}+\ep$ favors the decoding of $U_{1c}$ at receiver~2 while
slightly degrading the rate of user~1 in~\reff{eq:scheme X2 U1c DMC R1}.
In particular, the achievable scheme in Section \ref{sec:scheme E} for $\la=\la_{\rm Costa \ 1}+\ep$ achieves:
\begin{align*}
&R_1\leq \reff{eq:scheme X2 U1c R1} = \log \lb 1+\al P_1 \rb+
\\&-\log \lb 1 +\f{(\al P_1+1)^2 P_2 }{\al P_1 (1+P_1+|a|^2 P_2+2\Re\{a\} \sqrt{\alb P_1 P_2})}|\ep|^2 \rb,
\\
&R_2\leq \reff{eq:scheme X2 U1c R2}
= \log \lb 1 + |b|^2 P_1+P_2+ 2\sqrt{\alb |b|^2P_1 P_2}\rb+
\\&+\log\lb \f{ 1}{1+\al |b|^2P_1} +\right.\\&\left.
+\f{(\al |b|^2P_1+1) P_2 }{\al P_1 (1+|b|^2 P_1+P_2+ 2\sqrt{\alb |b|^2P_1 P_2}))}
|\ep+\Delta_\la|^2 \rb,
\end{align*}
where $\Delta_\la=\la_{\rm Costa \ 1}-\la_{\rm Costa \ 2}.$
Let  $\ep$ be small and with the same phase as $\Delta_\la$.
From the above expressions, we see that the  decrease of $R_1$
%\reff{eq:scheme X2 U1c R1} in $\ep$
is of the order $\Ocal(|\ep|^2)$, while the increase of $R_2$
%\reff{eq:scheme X2 U1c R2} in $\ep$
is of the order $ \Ocal(|\ep|)$.
%, that is, the rate increase for user~2 compensate for the rate decrease of user~1.
{\em This demonstrates how one can trade residual interference at the cognitive receiver for rate improvement at the primary receiver,}
which makes use of $U_{1c}$ to decode its own message.

\subsection{On the limits of partial interference ``pre-cancellation''}

A natural question at this point is whether it is possible to achieve capacity by using
the strategy in Section \ref{sec:scheme E} with $\la = \la_{\rm Costa \ 1}+\ep$.
Unfortunately, the answer is negative:

\begin{lem}
\label{thm:no more capacity}
The achievability of the outer bound of Th.\ref{thm:unifying outer bound} using the achievable scheme of Section~\ref{sec:scheme E} can be shown only in the ``primary decodes cognitive'' regime of Th.\ref{thm:new capacity result}.
\end{lem}

\begin{proof}
This result is shown by observing that only one choice of $\al$ and $\la$ in \reff{eq:scheme X2 U1c} achieves both the sum rate bound in \reff{eq:outer bound CIFC R1+R2} and the $R_1$ bound in \reff{eq:outer bound CIFC R1}--the choice which corresponds to the ``primary decodes cognitive'' regime of Th.\ref{thm:new capacity result}.
%Such choice corresponds to the choice in the ``primary decodes cognitive" regime of Theorem \ref{thm:new capacity result} and thus no other capacity result can be derived. %from this inner and outer bounds.
%matches the sum rate bound and the $R_1$ bound in the inner and outer bound.
%
To distinguish the parameters $\al$ in the inner and outer bound, let $\al^{\rm(out)}$ be the $\al$ parameter in the outer bound in~\reff{eq:outer bound CIFC} and $\al^{\rm(in)}$ be the $\al$ parameter in~\reff{eq:scheme X2 U1c}.
%
%Consider the problem of showing the achievability of a point in the outer bound for some $\al^{\rm(out)}$.

%To show the achievability of the outer bound of Th.\ref{thm:unifying outer bound}, we have to show the achievability for each $\al^{\rm(out)} \in [0,1]$.
Consider  the region \reff{eq:outer bound CIFC} for a fixed $\al^{\rm(out)}$.
%
%We first notice that no assignment of $\al^{\rm(in)}$ other than $\al^{\rm(in)} \leq \al^{\rm(out)}$ can be used to achieve the outer bound.
We first notice that to achieve the sum rate outer bound in \reff{eq:outer bound CIFC R1+R2} %for a fixed $\al^{\rm(out)}$,
we have to pick the $\al^{\rm(in)}$ in \reff{eq:scheme X2 U1c} such that  $\al^{\rm(in)} \leq \al^{\rm(out)}$ (because the expressions are monotonically decreasing in $\al$).
%The two regions \reff{eq:outer bound CIFC} and \reff{eq:scheme X2 U1c} have the same sum rate expression in $\al$ which is monotonically decreasing in $\al$  and thus only $\al^{\rm(in)} \leq \al^{\rm(out)}$ ensures that the sum rate bound for the inner bound is larger that the one of the outer bound. %
%For this reason to achieve \reff{eq:outer bound CIFC R1}+\reff{eq:outer bound CIFC R2} we must have  $\al^{\rm(in)} \leq  \al^{\rm(out)}$.
%
On the other hand, the maximum rate $R_1$ in \reff{eq:scheme X2 U1c R1} for $\al^{\rm(in)} \leq \al^{\rm(out)}$ is always smaller than the outer bound for $R_1$ in \reff{eq:outer bound CIFC R1},
with equality only if $\al^{\rm(in)}=\al^{\rm(out)}$ and $\la=\la_{\rm Costa \ 1}$.
%, that is:
%$$
%\max_{\al^{\rm(in)} \leq  \al^{\rm(out)}, \la \in \Cbb }  \reff{eq:scheme X2 U1c R1} \leq \Ccal(\al^{\rm(in)} P_1)=\reff{eq:outer bound CIFC R1}.
%$$
%The maximum is achieved for $\al^{\rm(in)}=\al^{\rm(out)}$ and $\la=\la_{\rm Costa \ 1}$.
So to achieve both the sum rate outer bound and the $R_1$ rate outer bound we must have  $\al^{\rm(in)}=\al^{\rm(out)}$ and $\la=\la_{\rm Costa \ 1}$, which are the specific assignments considered  in Theorem \ref{thm:new capacity result}.
%
%Since the sum rate bound is the same in the two regions, $\al$ in \reff{eq:scheme X2 U1c} must be less or equal than the bound in \reff{eq:outer bound CIFC}.
%The bound \ref{eq:scheme X2 U1c R1} is a moronically increasing function of $\al$ and maximum in $\la=\la_{\rm Costa}$ equal to \reff{eq:outer bound CIFC R1}. Therefore by decreasing $\al$ in  \reff{eq:scheme X2 U1c}
%is not possible to achieve \reff{eq:outer bound CIFC R1}.
\end{proof}

\subsection{Further considerations on the utility of partial interference ``pre-cancellation''}

In spite of the result of Lemma~\ref{thm:no more capacity}, we next show that the largest
inner bound region with the scheme of Section \ref{sec:scheme E} is achieved
by $\la \not= \la_{\rm Costa \ 1}$.
Finding the $\la$ that minimizes the distance between the inner and
outer bounds is analytically too involved.
For this reason we consider the simpler problem of determining the value
of $\la$ that maximizes the sum rate. %in \reff{eq:scheme X2 U1c R1+R2} for a given $\al$.

\begin{lem}
\label{thm:max sum rate}
When $|b|>1$ and the ``very strong interference'' condition is not satisfied,
setting $\la$ to a solutions of
\begin{align}
&-P_2 (|b|^2 H_2^{-1}-H_1^{-1}+\f{H_2^{-1}-H_1^{-1}}{\al P_1}) |\la|^2+ \nonumber\\
&2 \lb \sqrt{\alb P_1 P_2}(|b|^2 H_2^{-1}-H_1^{-1})+P_2 ( |b| H_2^{-1}-\Re\{a\} H_1^{-1} )\rb \Re\{\la\}  \nonumber\\
&\al P_1 (|b|^2 H_2^{-1}-H_1^{-1}) =0.
%&2 \lb \sqrt{\alb P_1 P_2}(|b|^2 H_2^{-1}-H_1^{-1})+P_2 ( |b| H_2^{-1}-\Re\{a\} H_1^{-1} )\rb x \nonumber \\
%&-P_2 \lb |b|^2 H_2^{-1}-H_1^{-1}+\f{H_2^{-1}-H_1^{-1}}{\al P_1}\rb x^2
%\nonumber \\
%&+\al P_1 (|b|^2 H_2^{-1}-H_1^{-1})=0
\label{eq:max sum rate}
\end{align}
for
\begin{align*}
H_1 &= \EE[|Y_1|^2] = 1 +|a|^2 P_2 + P_1 +2 \Re\{a\} \sqrt{\alb P_1 P_2}, \\
H_2 &= \EE[|Y_2|^2] = 1 + P_2 +|b|^2 P_1 +2 \sqrt{\alb |b|^2 P_1 P_2},
%\label{eq:h1 and h2}
\end{align*}
yields the largest achievable sum rate in the scheme of Section \ref{sec:scheme E} for a given $\al\in[0,1]$.
\end{lem}

\begin{proof}
This result is established by observing that \reff{eq:scheme X2 U1c R1+R2} is
the maximum achievable sum rate for a fixed $\al$.  Then, the $\la$ that maximizes
the sum rate must satisfy
$\reff{eq:scheme X2 U1c R1}+\reff{eq:scheme X2 U1c R2} =
\reff{eq:scheme X2 U1c R1+R2}$, that is,
\begin{align}
&f\Big(a+\sqrt{\f {\alb P_1}{P_2}},  1; \ \la\Big)
%\nonumber\\
%&
=
f\Big(\frac{1}{|b|}+\sqrt{\f {\alb P_1}{P_2}},\frac{1}{|b|^2}; \ \la\Big).
\label{eq:la for max sumrate in ach}
\end{align}
After some algebra, we rewrite the condition in~\reff{eq:la for max sumrate in ach}
as in~\reff{eq:max sum rate}.

%{\footnotesize
%\begin{align}
%&-P_2 (|b|^2 H_2^{-1}-H_1^{-1}+\f{H_2^{-1}-H_1^{-1}}{\al P_1}) |\la|^2+ \nonumber\\
%&2 \lb \sqrt{\alb P_1 P_2}(|b|^2 H_2^{-1}-H_1^{-1})+P_2 ( |b| H_2^{-1}-\Re\{a\} H_1^{-1} )\rb \Re\{\la\}  \nonumber\\
%&\al P_1 (|b|^2 H_2^{-1}-H_1^{-1})=0.
%\label{eq:complex max sum rate}
%\end{align}
%}\normalsize
%for $H_1$ and $H_2$ defined as in \reff{eq:h1 and h2}.

Notice that there exist real-valued solutions (i.e., $\la\in\RR$) for in~\reff{eq:max sum rate}.
Indeed, when the ``very strong interference'' condition is not verified, and $|b| \geq 1$,
we have $H_1 \leq H_2$ for all $\al\in[0,1]$, and thus it follows that:
the coefficient of $|\la|^2$ is negative,
the coefficient of $\Re \{\la\}$ is positive,
and the constant term is positive.
%The imaginary part of $|\la|$ appears always a negative contribution to the LHS of \reff{eq:complex max sum rate}: for this reason setting $\Im\{\la\}=0$ minimizes  $|\la|$.
%\nrd{???What does that mean???}
%For this reason $\phi_{\la}=0$ is the optimal choice for the phase of $\la$.
By choosing $\Im\{\la\}=0$,
the equation  \reff{eq:max sum rate} reduces to
the quadratic function in $\Re \{\la\}$,
which has positive definite determinant and thus has at least one real-valued
solution. %exists and hence that sum rate optimality can be always granted.
\end{proof}

\subsection{Numerical results}
\label{sec:numerical}
For the numerical results in the following we restrict ourselves
to real-valued input/output G-CIFC so as to reduce the dimensionality
of the search space for the optimal parameter values.

Although choosing $\la$ according to Lemma \ref{thm:max sum rate} does not guarantee
achieving the largest inner bound, we next show by numerical evaluations
that significant rate improvements can be obtained
when compared to choosing $\la=\la_{\rm Costa \ 1}$.
Fig.~\ref{fig:MaricChannelAchievabilityV4} shows the position of the point
\[
D(\la)=\Big(\reff{eq:scheme X2 U1c R1}, \min \{\reff{eq:scheme X2 U1c R2}, \reff{eq:scheme X2 U1c R1+R2}-\reff{eq:scheme X2 U1c R1}\}\Big)
\]
in the range $\la \in [0,2 \la_{\rm Costa \ 1}]$, for a fixed $\al^{\rm(in)}$,
together the outer bound point $C$ for $\al^{\rm(out)}=\al^{\rm(in)}$.
%{\red DT: 	ADD ALSO THE WHOLE STRONG INTERF OUTER BOUND.
%ADD ALSO MAX SUM RATE LINES--IF FIG DOES NOT GETTOO CROWDED.
%AND THE VALUE OF ALPHA IN THE CAPTION OF FIG.}
%\nbl{I would also suggest drawing on the "distance" that is being minimized in Lemma V.3.}
%
Under the ``primary decodes cognitive'' condition, $D(\la_{\rm Costa \ 1})=C$
for every $\al\in[0,1]$.
%this curve touches the outer bound point in $\la=\la_{\rm Costa \ 1}$ for every $\al$.
However, here we show a channel where the condition in \reff{eq:capacity condition 1} is not satisfied.
In this case the choice $\la=\la_{\rm Costa \ 1}$  minimizes the distance of the $R_1$-coordinate between $D$ and $C$, but it does not minimize the Euclidean distance between the two points.
The choice of $\la$ as in Lemma \ref{thm:max sum rate} minimizes the distance between the sum rate inner and outer bounds.

\begin{figure}
\centering
\includegraphics[width=9 cm ]{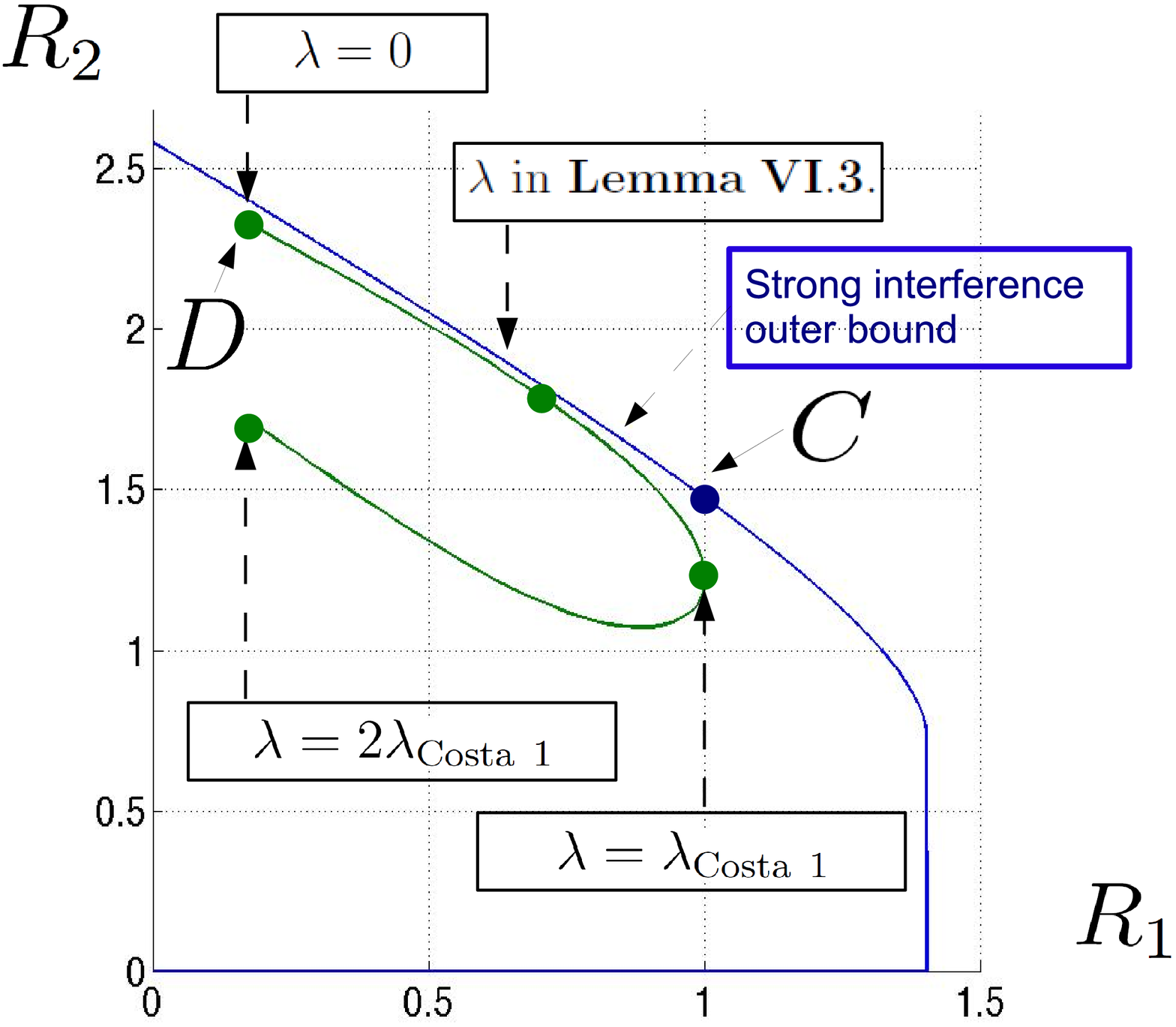}
  \vspace{-1. cm}
\caption{A plot of the points $C(\la)$ and $D(\la)$ for $\al=0.5$,
together with the ``strong interference'' bound for the
real-valued G-CIFC with parameters $P_1=P_2=6, b=\sqrt{2}$ and  $a=\sqrt{0.3}$.}\label{fig:MaricChannelAchievabilityV4}
\end{figure}

The following is another example to show that partial interference ``pre-cancellation''
($\la\not=\la_{\rm Costa \ 1}$) yields an achievable region larger than with
perfect interference ``pre-cancellation'' ($\la=\la_{\rm Costa \ 1}$).
Fig.~\ref{fig:SumOptimal} illustrates the  achievable region obtained with the choice of $\la$ as in Lemma \ref{thm:max sum rate}.  For this choice of parameters, the region obtained by considering any  $\la \in [0,2 \la_{\rm Costa \ 1}]$ coincides with the region obtained by choosing $\la$ as in Lemma \ref{thm:max sum rate}.
The achievable rate region for $\la=\la_{\rm Costa \ 1 }$ is also provided for reference.
Although numerical evaluations show that this is not the case in general, it is interesting to note that this choice  of $\la$ has the same achievable rate region as the one obtained using any (optimal) values of $\la$. %the parameters for certain channels. % in a certain parameter range.
% from the simulation is clear that the sum rate optimality of $\la$ yields the largest inner bound for the scheme of Section \ref{sec:scheme E} over any $\al$  and $\la$.  Simulation also shows that the sum rate optimality of $\la$ does not always yield the largest achievable region.

\begin{figure}
\centering
\includegraphics[width=9.0 cm ]{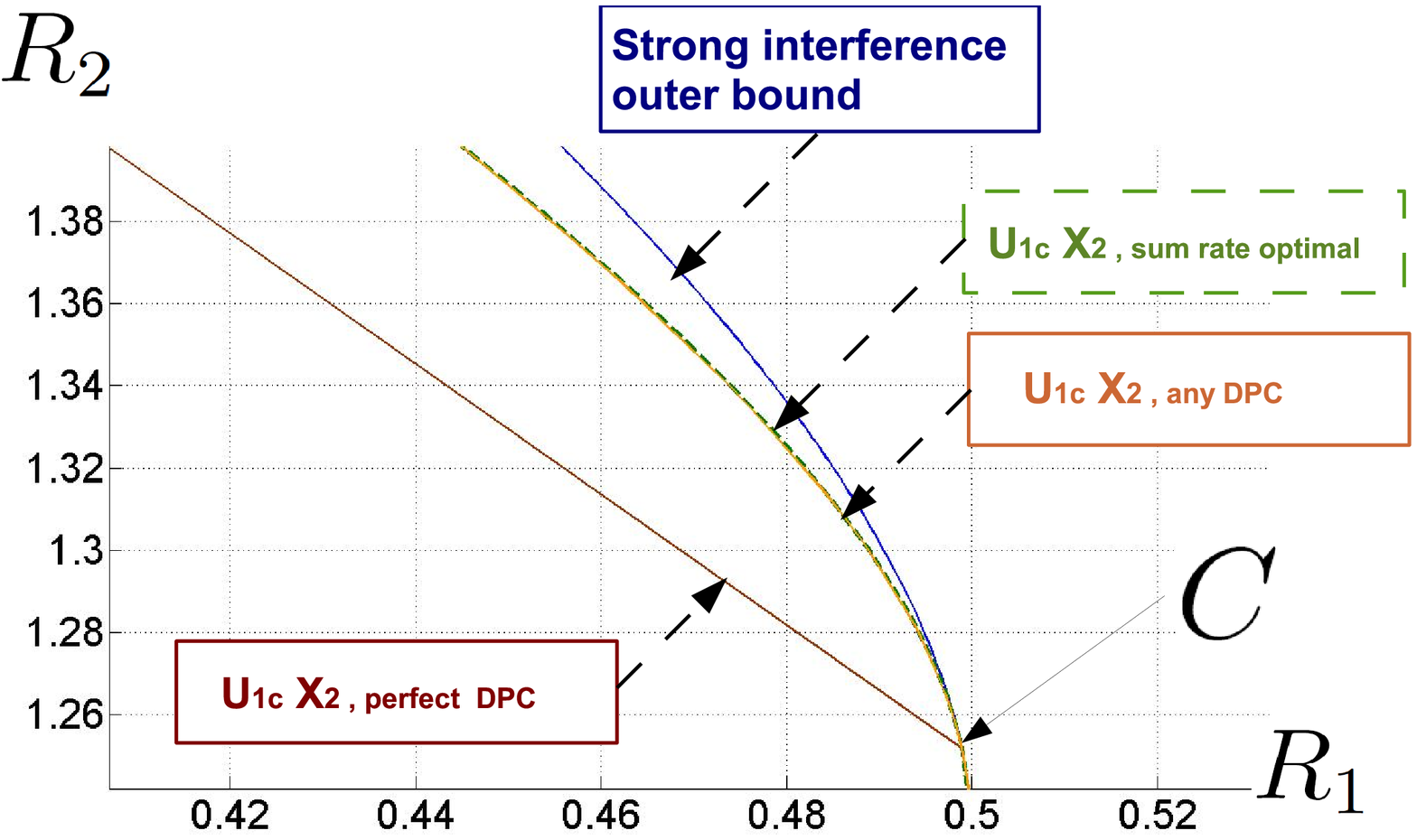}
\vspace{-3 cm}
\caption{A portion of the achievable region in~\reff{eq:scheme X2 U1c} for
$\la=\la_{\rm Costa}$ (labeled as ``perfect DPC''),
the optimal $\la \in [0,2 \la_{\rm Costa}]$ (labeled as ``any DPC''),
and $\la$ chosen as  in Lemma \ref{thm:max sum rate} (labeled as ``sum rate optimal'')
for the real-valued G-CIFC with parameters
$P_1=P_2=6$, $b=3$  and $a=2$.}
\vspace{-.5 cm}
\label{fig:SumOptimal}
\end{figure}

%{\red DT: in Fig.9 remove ``(numerical simulation)''.}

\textbf{Acknowledgment} The authors would like to thank Prof. Shlomo Shamai for insightful discussions on the problem during IZS 2010 (International Zurich Seminar on Communications, March 3-5, 2010, Zurich, Switzerland.)

\section{Conclusion and Future Work}
\label{sec:Conclusion and Future Work}

In this paper we presented a new capacity result for the Gaussian cognitive interference channel in a subset of the channel parameter space which we term the ``primary decodes cognitive" regime.
We derived an additive and a multiplicative approximate capacity result that provides important insight on the fundamental features of the capacity achieving scheme.
In particular, we show how pre-coding against the  interference for one user can be used to boost the rate of the other user.
Numerical results further verify that significant  rate improvements can be obtained.
%
%Despite this new result,
While this result extends the parameter regimes in which capacity for the Gaussian cognitive interference channel is known, it still remains unknown in general.
%The outer bound does not reduce to the capacity of the broadcast channel when the power at the primary receiver is set to zero.
%This observation suggests that the strong interference outer bound in not tight on the whole parameter region. For this reason we are currently investigating outer bounds based on the broadcasting nature of the cognitive interference channel.
%It is possible to show the outer bound available for region where capacity is unknown is not tight in general.
The achievable region of \cite{rini2009state} provides a comprehensive inner bound that can potentially yield new capacity results.
%The peculiarity of the capacity achieving scheme in the ``primary decoding cognitive" regime can be combined with other encoding strategies to obtain a potentially larger achievable region.
Only some specific choices of parameters for this region have been considered so far and we thus expect that additional results may be derived from this region.
On the other hand, tighter outer bounds for the capacity region are also probably necessary.
When the Gaussian cognitive interference channel reduces to a broadcast channel, the outer bound is not tight, suggesting  that the outer bound may be loose in an even larger region.

%For these reasons we are currently developing new outer bounds that may give rise to new capacity results.
%
%\nrd{These last 2 sentences really have little to do with the rest - are they the best "future work" statements that we could think of? Nothing more related to the scheme?}
%{\red
%to complete
%}

\bibliographystyle{IEEEtran}
\bibliography{steBib1}

\end{document}